\documentclass[11pt]{article}

\usepackage[T1]{fontenc}
\usepackage{textcomp}
\usepackage{palatino}
\usepackage{mathpazo}
\usepackage{stmaryrd}

\usepackage{hyperref}
\hypersetup{pdfpagemode=UseNone}

\usepackage{fullpage}
\usepackage{amsfonts}
\usepackage{amssymb}
\usepackage{amsmath}
\usepackage{latexsym}
\usepackage{amsthm}
\usepackage{eepic}
\usepackage{sectsty}
\usepackage[usenames]{color}

\newtheorem{theorem}{Theorem}
\newtheorem{lemma}[theorem]{Lemma}

\newtheorem{cor}[theorem]{Corollary}
\theoremstyle{definition}
\newtheorem{definition}[theorem]{Definition}

\newtheorem{fact}[theorem]{Fact}

\newenvironment{proofof}[1]{\noindent{\bf Proof of #1:}}{\qed\\}

\newcommand{\tinyspace}{\mspace{1mu}}

\newcommand{\norm}[1]{\left\lVert\tinyspace#1\tinyspace\right\rVert}

\newcommand{\defeq}{\stackrel{\smash{\text{\tiny def}}}{=}}
\newcommand{\tr}{\operatorname{Tr}}

\def\({\left(}
\def\){\right)}
\def\I{\mathbb{1}}

\newcommand{\imax}{\mathrm{Imax}}
\newcommand{\imin}{\mathrm{Imin}}
\newcommand{\local}{\mathrm{local}}
\newcommand{\myglobal}{\mathrm{global}}

\def\real{\mathbb{R}}

\def\O{\mathcal{O}}

\def\ve{{\varepsilon}}

\def\polylog{\mathrm{polylog}}

\newcommand{\suppress}[1]{}

\newcommand{\seq} [1] {#1_1,\ldots,#1_m}

\title{A parallel approximation algorithm for mixed packing and covering semidefinite programs}

\author{
Rahul Jain\thanks{
Centre for Quantum Technologies and Department of Computer Science,
National University of Singapore,  Block S15, 3 Science Drive~2,
Singapore 11754.
Email: \texttt{rahul@comp.nus.edu.sg}.
} \\
National U.\ Singapore
\and
Penghui Yao\thanks{
Centre for Quantum Technologies and Department of Computer Science,
National University of Singapore,  Block S15, 3 Science Drive~2,
Singapore 11754.
Email: \texttt{pyao@nus.edu.sg}.
} \\
National U.\ Singapore
}

\date{January 28, 2012}

\begin{document}

\maketitle

\abstract{We present a parallel approximation algorithm for a class of {\em mixed packing and covering} semidefinite programs which generalize on the class of positive semidefinite programs as considered by Jain and Yao~\cite{JainY11}. As a corollary we get a faster approximation algorithm for positive semidefinite programs with better dependence of the parallel running time on the approximation factor, as compared to that of Jain and Yao~\cite{JainY11}. Our algorithm and analysis is on similar lines as that of Young~\cite{Young01} who considered analogous linear programs.}

%\thispagestyle{empty}

%\newpage

\setcounter{page}{1}

\section{Introduction}

Fast parallel  approximation algorithms for semidefinite programs have been the focus of study of many recent works (e.g.~\cite{AroraHK05, AroraK07, Kale07, JainW09, JainUW09, JainJUW10}) and have resulted in many interesting applications including the well known $\mathsf{QIP=PSPACE}$~\cite{JainJUW10} result. In many of the previous works, the running time of the algorithms had polylog dependence on the size of the input program but in addition also had polynomial dependence of some {\em width} parameter (which varied for different algorithms). Sometimes (for specific instances of input programs) the width parameter could be as large as the size of the program  making it an important bottleneck. Recently Jain and Yao~\cite{JainY11} presented a fast parallel approximation algorithm for an important subclass of semidefinite programs, called as {\em positive semidefinite programs}, and their algorithm had no dependence on any width parameter. Their algorithm was inspired by an algorithm by Luby and Nisan~\cite{LubyN93} for {\em positive linear programs}. In this work we consider a more general {\em mixed packing and covering} optimization problem. We first consider the following feasibility task ${\bf Q1}$. 

\vspace{0.1in}

\noindent {\bf Q1:} Given $n\times n$ positive semidefinite matrices $\seq{P}, P$ and non-negative diagonal matrices $\seq{C}, C$  and $\ve\in(0,1)$, find an $\ve$-{\it approximate feasible} vector $x\geq 0$ such that  (while comparing matrices we let $\geq, \leq$ represent the L\"{o}wner order),
$$\sum_{i=1}^m x_iP_i\leq(1+\ve)P \quad \mbox{  and   } \quad \sum_{i=1}^m x_iC_i\geq C$$ or show that the following is infeasible  
$$\sum_{i=1}^m x_iP_i\leq P \quad \mbox{ and  } \quad \sum_{i=1}^m x_iC_i\geq C \enspace .$$ 

We present an algorithm for ${\bf Q1}$ running in parallel time $\polylog (n,m) \cdot \frac{1}{\ve^4} \cdot \log \frac{1}{\ve}$. Using this and standard binary search, a multiplicative $(1 - \ve)$ approximate solution can be obtained for the following optimization task ${\bf Q2}$ in parallel time $\polylog (n,m) \cdot \frac{1}{\ve^4} \cdot \log \frac{1}{\ve}$.  

\vspace{0.1in}

\noindent {\bf Q2:} Given $n\times n$ positive semidefinite matrices $\seq{P},  P$ and non-negative diagonal matrices $\seq{C}, C$,
\begin{center}
  \begin{minipage}{2in}
    \begin{align*}
      \text{maximize:}\quad & \gamma \\
      \text{subject to:}\quad & \sum_{i=1}^m x_iP_i\leq  P \\
      & \sum_{i=1}^m x_iC_i\geq \gamma C \\		
      & \forall i \in [m]: x_i \geq 0.
    \end{align*}
  \end{minipage}
\end{center}      
The following special case of {\bf Q2} is referred to as a positive semidefinite program.  

\vspace{0.1in}

\noindent {\bf Q3:} Given $n\times n$ positive semidefinite matrices $\seq{P},P$ and non-negative scalars $c_1, \ldots, c_m$,
\begin{center}
  \begin{minipage}{2in}
    \begin{align*}
      \text{maximize:}\quad & \sum_{i=1}^m x_ic_i \\\\
      \text{subject to:}\quad & \sum_{i=1}^m x_iP_i\leq P \\		
      & \forall i \in [m]: x_i \geq 0.
    \end{align*}
  \end{minipage}
\end{center}      
Our algorithm for {\bf Q1} and its analysis is on similar lines as the algorithm and analysis of Young~\cite{Young01} who had considered analogous questions for linear programs. As a corollary we get an algorithm for approximating positive semidefinite programs ({\bf Q3}) with better dependence of the parallel running time on $\ve$ as compared to that of Jain and Yao~\cite{JainY11} (and arguably with simpler analysis). Very recently, in an independent work, Peng and Tangwongsan~\cite{PengT11} also presented a fast parallel algorithm for positive semidefinite programs. Their work is also inspired by Young~\cite{Young01}.

\section{Algorithm and analysis}
We mention without elaborating that using standard arguments the feasibility question {\bf Q1} can be easily transformed, in parallel time $\polylog(mn)$, to the special case when $P$ and $C$ are identity matrices and we consider this special case from now on. Our algorithm is presented in Figure~\ref{fig:alg} ~ .

\subsection*{Idea of the algorithm} The algorithm starts with an initial value for $x$ such that  $\sum_{i=1}^m x_i P_i \leq \I$. It makes increments to the vector $x$ such that with each increment, the increase in $\norm{\sum_{i=1}^m x_i P_i}$ is not more than $(1+\O(\ve))$ times the increase in the minimum eigenvalue of $\sum_{i=1}^m x_i C_i$. We argue that it is always possible to increment $x$ in this manner if the input instance is feasible, hence the algorithm outputs $\mathsf{infeasible}$ if it cannot find such an increment to $x$. The algorithm stops when  the minimum eigenvalue of $\sum_{i=1}^m x_i C_i$ has exceeded $1$. Due to our condition on the increments, at the end  of the algorithm we also have $\sum_{i=1}^m x_i P_i \leq (1+\O(\ve)) \I$. We obtain handle on the largest and smallest eigenvalues of concerned matrices via their {\em soft} 
versions, which are more easily handled functions of those matrices (see definition in the next section). 

\begin{figure}[!ht]

\noindent\hrulefill

{
\small

\noindent {\bf Input :} $n\times n$ positive semidefinite matrices $\seq{P}$, non-negative diagonal matrices $\seq{C}$, and error parameter $\ve \in (0,1)$.

\medskip

\noindent {\bf Output :} Either $\mathsf{infeasible}$, which means there is no $x$ such that ($\I$ is the identity matrix),
$$\sum_{i=1}^m x_iP_i\leq \I \quad \mbox{ and } \quad \sum_{i=1}^m x_iC_i\geq \I \enspace.$$ 
OR an $x^* \in \real^m$ such that 
$$\sum_{i=1}^m x_i^* P_i\leq (1+9\ve) \I \quad \mbox{ and } \quad \sum_{i=1}^m x_i ^* C_i\geq \I \enspace .$$ 

\begin{enumerate}

\item Set $x_j = \frac{1}{m\norm{P_j}}$.

\item Set $N = \frac{1}{\ve}\left(\norm{\sum_{i=1}^m x_iP_i}+2\ln n+\ln m\right)$.

\item While $\lambda_{\min}(\sum_{i=1}^mx_iC_i)<N$ ($\lambda_{\min}$ represents minimum eigenvalue), do

\begin{enumerate}

\item Set
\begin{align*}
\local_j(x) &= \frac{\tr(\exp(\sum_{i=1}^m x_iP_i)\cdot P_j)}{\tr(\exp(-\sum_{i=1}^m x_iC_i)\cdot C_j)} \quad \text{ and } \\
\myglobal(x) &= \frac{\tr{\exp(\sum_{i=1}^m x_iP_i)}}{\tr(\exp(-\sum_{i=1}^mx_iC_i))} .
\end{align*}

\item If $g$ is not yet set or $\min_j\{\local_j(x)\}>g(1+\ve)$,  set $g = \myglobal(x)$.

\item If $\min_j \{\local_j(x) \}> \myglobal(x)$ , return $\mathsf{infeasible}$.

\item For all $j \in [m]$, set 
$C_j = \Pi_j  \cdot C_j \cdot \Pi_j$, where $\Pi_j$ is the projection onto the eigenspace of $\sum_{i=1}^mx_iC_i$ with eigenvalues at most $N$.

\item Choose increment vector $\alpha\geq0$ and scalar $\delta>0$ such that
$$\forall j: ~ \alpha_j = x_j \delta \text{ if } \local_j(x)\leq g(1+\ve), \text{ else } \alpha_j = 0, \text{ and } $$
$$\max\{\norm{\sum_{i=1}^m \alpha_iP_i},\norm{\sum_{i=1}^m\alpha_iC_i}\}=\ve .$$

\item Set $x = x+\alpha$.
\end{enumerate}

\item Return $x^* = x/N$.

\end{enumerate}

}
\noindent\hrulefill

\caption{Algorithm} \label{fig:alg}

\end{figure}

\subsection*{Correctness analysis} 

We begin with the definitions of soft maximum and minimum eigenvalues of a positive semidefinite matrix $A$. They are inspired by analogous definitions made in Young~\cite{Young01} in the context of vectors.
\begin{definition}
For positive semidefinite matrix $A$, define 
$$\imax(A)\defeq\ln\tr\exp(A),$$
and
$$\imin(A)\defeq-\ln\tr\exp(-A).$$
Note that $\imax(A)\geq\norm{A}$ and $\imin(A)\leq\lambda_{\min}(A)$, where $\lambda_{\min}(A)$ is the minimum eigenvalue of $A$.
\end{definition}
We show the following lemma in the appendix, which  shows that if a small increment is made in the vector $x$, then changes in $\imax(\sum_{j=1}^m x_j A_j )$ and $\imin(\sum_{j=1}^m x_j A_j )$  can be bounded appropriately.
\begin{lemma}\label{lem:smooth}
Let $\seq{A}$ be positive semidefinite matrices and let $x\geq 0,\alpha\geq 0$ be vectors in $\real^m$. If $\norm{\sum_{i=1}^m\alpha_iA_i}\leq\ve \leq 1$, then
$$\imax(\sum_{j=1}^m (x_j + \alpha_j) A_j )-\imax(\sum_{j=1}^m x_j  A_j)\leq \frac{(1+\ve)}{\tr(\exp(\sum_{i=1}^mx_iA_i))}\sum_{j=1}^m\alpha_j \tr(\exp(\sum_{i=1}^mx_iA_i)A_j),$$
and
$$\imin(\sum_{j=1}^m (x_j + \alpha_j)  A_j)-\imin(\sum_{j=1}^m x_j  A_j)\geq\frac{(1-\ve/2)}{\tr(\exp(-\sum_{i=1}^mx_iA_i))}\sum_{j=1}^m\alpha_j\tr(\exp(-\sum_{i=1}^mx_iA_i)A_j).$$
\end{lemma}

\begin{lemma}\label{lem:gincrease}
At step 3(e) of the algorithm, for any $j$ with $\alpha_j > 0$ we have,
$$ \frac{\tr(\exp(\sum_{i=1}^m x_iP_i)\cdot P_j)}{ \tr (\exp(\sum_{i=1}^m x_iP_i)) } \leq (1+\ve) \frac{\tr(\exp(-\sum_{i=1}^m x_iC_i)\cdot C_j)}{\tr(\exp(-\sum_{i=1}^mx_iC_i))} .$$ 
\end{lemma}

\begin{proof}
Consider any execution of step 3(e) of the algorithm. Fix $j$ such $\alpha_j > 0$. Note that,
$$\frac{\local_j(x)}{\myglobal(x)} = \frac{\tr(\exp(\sum_{i=1}^m x_iP_i)\cdot P_j)\cdot\tr(\exp(-\sum_{i=1}^mx_iC_i))}{ \tr (\exp(\sum_{i=1}^m x_iP_i)) \cdot \tr(\exp(-\sum_{i=1}^m x_iC_i)\cdot C_j)}.$$
We will show that $\myglobal(x) \geq g$ throughout the algorithm and this will show the desired since that $\local_j(x) \leq (1+\ve) g \leq (1+\ve)\myglobal(x)$. 

At step 3(b) of the algorithm, $g$ can be equal to $\myglobal(x)$. Since $x$ never decreases during the algorithm, at step 3(a), $\myglobal(x)$ can only increase. At step 3(d),  the modification of $C_j$s only decreases $\tr(\exp(-\sum_{i=1}^mx_iC_i))$ and hence again $\myglobal(x)$ can only increase.  
\end{proof}

\begin{lemma}\label{lem:iminvsimax}
For each increment of $x$ at step 3(f) of the algorithm, 
$$\imax(\sum_{j=1}^m (x_j + \alpha_j) P_j )-\imax(\sum_{j=1}^m x_j  P_j) \leq (1+\ve)^3 \(\imin(\sum_{j=1}^m (x_j + \alpha_j) C_j )-\imin(\sum_{j=1}^m x_j  C_j) \) . $$ 
\end{lemma}

\begin{proof}
Consider,
\begin{align*}
\lefteqn{\imax(\sum_{j=1}^m (x_j + \alpha_j) P_j )-\imax(\sum_{j=1}^m x_j  P_j) } \\
&\leq \frac{(1+\ve)}{\tr(\exp(\sum_{i=1}^mx_iP_i))} \sum_{j=1}^m\alpha_j \tr(\exp(\sum_{i=1}^mx_iP_i)P_j) \quad \mbox{(from Lemma~\ref{lem:smooth})}\\
&\leq \frac{(1+\ve)^2}{\tr(\exp(-\sum_{i=1}^mx_iC_i))}  \sum_{j=1}^m\alpha_j \tr(\exp(-\sum_{i=1}^mx_iC_i)C_j) \quad \mbox{(from Lemma~\ref{lem:gincrease} and step 3(e) of the algorithm)} \\
& \leq \frac{(1+\ve)^2}{1-\ve/2}\(\imin(\sum_{j=1}^m (x_j + \alpha_j)  C_j)-\imin(\sum_{j=1}^m x_j  C_j) \)  \quad \mbox{(from Lemma~\ref{lem:smooth})}.
\end{align*}
This shows the desired.
\end{proof}

\begin{lemma}\label{lem:feasible}
If the input instance $\seq{P}, \seq{C}$ is feasible, that is there exists vector $y \in \real^m$ such that 
$$\sum_{i=1}^m y_iP_i\leq \I \quad \mbox{ and  } \quad \sum_{i=1}^m y_iC_i\geq \I \enspace ,$$ 
then always at step 3(c) of the algorithm, $\min_j \{  \local_j(x) \} \leq  \myglobal(x)$. Hence the algorithm will return some $x^*$.

If the algorithm outputs $\mathsf{infeasible}$, then the input instance is not feasible.
\end{lemma}

\begin{proof}
Consider some execution of step 3(c) of the algorithm. Let $C'_1, \ldots, C'_m$ be the current values of $\seq{C}$. Note that if the input is feasible with vector $y$, then we will also have
\begin{align*}
\frac{\tr(\exp(\sum_{i=1}^mx_iP_i)(\sum_{j=1}^my_jP_j))}{\tr(\exp(\sum_{i=1}^mx_iP_i))}\leq 1 \leq 
\frac{\tr(\exp(-\sum_{i=1}^mx_iC'_i)(\sum_{j=1}^my_jC'_j))}{\tr(\exp(-\sum_{i=1}^mx_iC'_i))} .
\end{align*}
Therefore there exists $j \in [m]$ such that 
$$ \frac{\tr(\exp(\sum_{i=1}^mx_iP_i)P_j)}{\tr(\exp(\sum_{i=1}^mx_iP_i))}\leq \frac{\tr(\exp(-\sum_{i=1}^mx_iC'_i)C'_j)}{\tr(\exp(-\sum_{i=1}^mx_iC'_i))} ,$$
and hence $\local_j(x) \leq \myglobal(x)$.

If the algorithm outputs $\mathsf{infeasible}$, then  at that point $\min_j \{  \local_j(x) \} >  \myglobal(x)$ and hence from the argument above $\seq{P}, C'_1, \ldots, C'_m$ is infeasible which in turn implies that $\seq{P}, \seq{C}$ is infeasible.
\end{proof}

\begin{lemma}\label{lem:optimal}
If the algorithm  returns some $x^*$, then 
$$\sum_{i=1}^mx_i^*P_i \leq (1+9\ve)\I \quad \mbox{ and } \quad \sum_{i=1}^mx_i^*C_i \geq \I .$$
\end{lemma}
\begin{proof}
Because of the condition of the while loop, it is clear that $\sum_{i=1}^mx_i^*C_i \geq \I$.

For $x\in \real^m$, define  $$\Phi(x) \defeq \imax(\sum_{j=1}^m x_j  P_j) - (1+\ve)^3 \cdot \imin(\sum_{j=1}^m x_j  C_j) .$$ 
Note that the update of $C_j$'s  at step 3(d) only increase $\imin(\sum_{j=1}^m x_j  C_j)$. Hence using Lemma \ref{lem:iminvsimax}, we conclude that $\Phi(x)$ is non-decreasing during the algorithm. At step 1 of the algorithm, 
\begin{align*}
\Phi(x) &\leq \imax(\sum_{j=1}^m x_j  P_j)  = \ln \tr (\exp(\sum_{i=1}^m x_iP_i )) \\
& \leq  \ln(n\exp(\norm{\sum_{i=1}^mx_iP_i})) \leq \ln(n\exp(\sum_{i=1}^m\norm{x_iP_i})) = \ln n +1.
\end{align*}
Hence just before the last increment, 
\begin{align*}
\norm{\sum_{i=1}^mx_iP_i} \leq \imax(\sum_{j=1}^m x_j  P_j) & \leq \Phi(x) +  (1+\ve)^3 \cdot \imin(\sum_{j=1}^m x_j  C_j) \\
& \leq \ln n + 1 +(1+\ve)^3 \cdot \imin(\sum_{j=1}^m x_j  C_j) \\
& \leq \ln n + 1 +(1+\ve)^3 \cdot \lambda_{\min}(\sum_{j=1}^m x_j  C_j) \\
& \leq \ln n + 1 + (1+\ve)^3 N\leq  (1+8\ve )N \enspace .
\end{align*}
In the last increment, because of the condition on step 3(e) of the algorithm, $\norm{\sum_{i=1}^mx_iP_i}$ increase by at most $\ve$. Hence $\sum_{i=1}^mx_i^*P_i \leq (1+9\ve)\I$. 
\end{proof}

\subsection*{Running time analysis}
\begin{lemma}\label{phase}
Assume that the algorithm does not return $\mathsf{infeasible}$ for some input instance. The number of times $g$ is increased at step 3(b) of the algorithm is  $\O(N/\ve)$. 
\end{lemma}

\begin{proof}
At the beginning of the algorithm $\tr(\exp(-\sum_{i=1}^mx_iC_i)) \leq n$ since each eigenvalue of $\exp(-\sum_{i=1}^mx_iC_i)$ is at most $1$. Also 
$\tr{\exp(\sum_{i=1}^m x_iP_i)} \geq 1 $.
Hence 
$$g = \myglobal(x) = \frac{\tr{\exp(\sum_{i=1}^m x_iP_i)}}{\tr(\exp(-\sum_{i=1}^mx_iC_i))}  \geq \frac{1}{n} \geq \frac{1}{\exp(N)}.$$ 
At the end of the algorithm $\lambda_{\min}({\sum_{i=1}^mx_iC_i)} \leq N + \ve \leq 2N$. Hence 
$$\tr(\exp(-\sum_{i=1}^mx_iC_i)) \geq  \norm{\exp(-\sum_{i=1}^mx_iC_i)} = \exp(-\lambda_{\min}(\sum_{i=1}^mx_iC_i)) \geq \exp(-2N ) . $$ 
Also  (using Lemma~\ref{lem:optimal})
$$\tr(\exp(\sum_{i=1}^mx_iP_i)) \leq  n \norm{\exp(\sum_{i=1}^mx_iP_i)} \leq n \exp((1+9\ve)  N ) \leq \exp(11 N) . $$ 
Hence $g \leq \myglobal(x)\leq \exp(13N) $. 

Whenever  $g$ is updated at step 3(b) of the algorithm, we have 
$$\myglobal(x) \geq \min_j\{\local_j(x)\}>(1+\ve)g$$ 
 just before the update  and 
$\myglobal(x) = g$ just after the update. Thus $g$ increases by at least $(1+\ve)$ multiplicative factor.  Hence the number of times $g$ increases is $\O(N/\ve)$.
\end{proof}

\begin{lemma}\label{lem:increment}
Assume that the algorithm does not return $\mathsf{infeasible}$ for some input instance. The number of iterations of the while loop in the algorithm for a fixed value of $g$ is $\O(N\log(mN)/\ve)$.
\end{lemma}

\begin{proof}
From Lemma~\ref{lem:optimal} and step 3(d) of the algorithm we have $\max\{\norm{\sum_{i=1}^mx_iP_i},\norm{\sum_{i=1}^mx_iC_i}\}=\O(N)$  throughout the algorithm. On the other hand we have $\max\{\norm{\sum_{i=1}^m \delta x_iP_i},\norm{\sum_{i=1}^m \delta x_iC_i}\}=\ve$ at step 3(e). Hence $\delta=\Omega(\ve/N)$ throughout the algorithm.

Let $x_j$ be increased in the last iteration of the while loop for a fixed value of $g$. Note that $x_j$ is initially $1/(m\norm{P_j})$ and at the end  $x_j$ is at most $10N/\norm{P_j}$ (since, using Lemma~\ref{lem:optimal}, $\norm{x_j P_j} \leq \norm{\sum_{i=1}^m x_j P_j} \leq 10 N$).  Hence the algorithm makes at most $\O(\log(mN)/\delta) = \O(N \log(mN)/\ve)$ increments for each $x_j$.
 
Note that $\local_j(x)$ only increases throughout the algorithm (easily seen for steps 3(d) and 3(e) of the algorithm). Hence since the last iteration of the while loop (for this fixed $g$)  increases $x_j$, it must be that each iteration of the while loop increases $x_j$.  Hence, the number of iterations of the while loop (for this fixed $g$) is $\O(N\log(mN)/\ve)$.
\end{proof}
We claim (without further justification) that each individual step in the algorithm can be performed in parallel time $\polylog(mn)$. Hence combining the above lemmas and using $N = \O(\frac{\ln(mn)}{\ve})$, we get
\begin{cor}
The parallel running time of the algorithm is upper bounded by $\polylog(mn) \cdot \frac{1}{\ve^4} \cdot \log \frac{1}{\ve}$.
\end{cor}

\bibliographystyle{plain}
\bibliography{ref}

\appendix
\section{Deferred proofs}

\begin{proofof}{Lemma~\ref{lem:smooth}}
We will use the following Golden-Thompson inequality.
\begin{fact}\label{lem:golden}
For Hermitian matrices $A,B : ~ \tr(\exp(A+B))\leq\tr\exp(A)\exp(B).$
\end{fact}
We will also need the following fact.
\begin{fact}\label{fact:2} Let $A$ be positive semidefinite with $\norm{A}\leq \ve \leq 1$. Then,
$$\exp(A) \leq \I + (1+\ve)A \quad \mbox{ and} \quad \exp(-A) \leq \I - (1-\ve/2)A .$$
\end{fact} 
Consider,
\begin{eqnarray*}
\lefteqn{\imax(\sum_{j=1}^m (x_j + \alpha_j) A_j )-\imax(\sum_{j=1}^m x_j  A_j) }\\
&=&\ln \(\frac{\tr\exp(\sum_{i=1}^m(x_i+\alpha_i)A_i)}{\tr\exp(\sum_{i=1}^mx_iA_i)} \)\\
&\leq&\ln \(\frac{\tr\exp(\sum_{i=1}^mx_iA_i)\exp(\sum_{j=1}^m\alpha_jA_j)}{\tr\exp(\sum_{i=1}^mx_iA_i)}\) \quad \mbox{(from Fact~\ref{lem:golden})}\\
&=&\ln\(\frac{\tr\exp(\sum_{i=1}^mx_iA_i)(\I + (1+ \ve)(\sum_{j=1}^m\alpha_jA_j))}{\tr\exp(\sum_{i=1}^mx_iA_i)}\) \quad \mbox{(from Fact~\ref{fact:2})}\\
&=& \ln\(1 + \frac{(1+\ve)\tr\exp(\sum_{i=1}^mx_iA_i)(\sum_{j=1}^m\alpha_jA_j)}{\tr\exp(\sum_{i=1}^mx_iA_i)}\)\\
&\leq&\frac{(1+\ve)\tr\exp(\sum_{i=1}^mx_iA_i)(\sum_{j=1}^m\alpha_jA_j)}{\tr\exp(\sum_{i=1}^mx_iA_i)} \quad \mbox{(since $\ln(1+a) \leq a$ for all real $a$)}
\end{eqnarray*}
The desired bound on $\imin(\sum_{j=1}^m (x_j + \alpha_j)  A_j)-\imin(\sum_{j=1}^m x_j  A_j)$ follows by analogous calculations.
\end{proofof}

\end{document}